\documentclass[conference]{IEEEtran}
\IEEEoverridecommandlockouts
\usepackage{cite,mathcomSTEv4}
\usepackage{amsmath,amssymb,amsfonts}
\usepackage{amsthm,enumitem}
\usepackage{algorithmic}
\usepackage{graphicx}
\usepackage{textcomp}
\usepackage{xcolor}
\def\BibTeX{{\rm B\kern-.05em{\sc i\kern-.025em b}\kern-.08em
    T\kern-.1667em\lower.7ex\hbox{E}\kern-.125emX}}

\usepackage{tikz,pgfplots}
\usepackage{tikz-3dplot}
 
\tdplotsetmaincoords{60}{115}
\pgfplotsset{compat=newest}

\usepackage{csquotes}

\pgfplotsset{compat=1.10}
\usepgfplotslibrary{fillbetween}

\usepackage{amsmath,amssymb,amsfonts}
\usepackage{algorithmic}
\usepackage{graphicx}
\usepackage{textcomp}
\usepackage{hyperref}
\usepackage{tikz,pgfplots}
\pgfplotsset{width=10cm,compat=1.9}
\usetikzlibrary{calc}

\usepackage{pgfplotstable}
\usetikzlibrary{positioning}
\usetikzlibrary{pgfplots.groupplots}
\usetikzlibrary{arrows}

\usepackage{xcolor,mathcomSTEv4}

\newcommand{\Prob}{\mathbb{P}}

\newcommand{\bbu}{\mathbf{u}}

\newcommand{\normal}{\mathcal{N}}
\newcommand{\one}{\mathbf{1}}

\newtheorem{lemma}{Lemma}
\newtheorem{theorem}{Theorem}

\begin{document}

\title{
Sharp asymptotics on the compression of \\
two-layer neural networks

\thanks{
%
The work of M.H.A. and R.P. was performed while at ISTA. M.H.A., S.B., M.M., and R.P. were partially supported by the 2019 Lopez-Loreta Prize. The work of S.R. was supported by the MOST grant 110-2221-E-A49 -052.}
}

\author{
Mohammad Hossein Amani,
\textsuperscript{1},
Simone Bombari,
\textsuperscript{2},
Marco Mondelli,
\textsuperscript{2},
Rattana Pukdee, and 
\textsuperscript{3},
Stefano Rini
\textsuperscript{4}.\\
\textsuperscript{1} 
{EPFL, Switzerland, email: \texttt{mh.amani1998@gmail.com}}\\
\textsuperscript{2}
{ISTA, Austria, email:  \texttt{\{simone.bombari marco.mondelli\}@ist.ac.at} } \\
\textsuperscript{3}
{CMU, USA, email: \texttt{rpukdee@andrew.cmu.edu} } \\
\textsuperscript{4}
{NYCU, Taiwan, email: \texttt{stefano.rini@nycu.edu.tw}}
}

\maketitle

\let\thefootnote\relax\footnotetext{To appear at the IEEE Information Theory Workshop (ITW) 2022, Mumbai, India.
}

\begin{abstract}
In this paper, we study the compression of a \emph{target} two-layer neural network with $N$ nodes into a \emph{compressed} network with $M<N$ nodes. 
More precisely, we consider the setting in which the weights of the target network are i.i.d. sub-Gaussian, and we minimize the population $L_2$ loss between the outputs of the target and of the compressed network, under the assumption of Gaussian inputs. 
By using tools from high-dimensional probability, we show that this non-convex problem can be simplified when the target network is sufficiently over-parameterized, and provide the error rate of this approximation as a function of the input dimension and $N$.
In this mean-field limit,
the simplified objective, as well as the optimal weights of the compressed network, does not depend on the realization of the target network, but only on expected scaling factors.
%
Furthermore, for networks with ReLU activation, we conjecture that the optimum of the simplified optimization problem is achieved by taking weights on the Equiangular Tight Frame (ETF), while the scaling of the weights and the orientation of the ETF depend on the parameters of the target network.
Numerical evidence is provided to support this conjecture. 
%
\end{abstract}

\section{Introduction}
The compression of Deep Neural Networks (DNNs) plays a crucial role in many modern-day applications in which inference has to be performed at minimal computational cost.
This occurs, for instance, in the IoT paradigm: here, classic DNN tasks (e.g., image classification, speech recognition, and anomaly detection) have to be implemented over diverse platforms with limited computational capabilities. 
In this scenario, large DNNs trained over vast datasets have to be severely compressed to fit the current hardware limitations.
Several empirical studies have shown that over-parameterized DNNs are often highly redundant and, hence, 
smaller sub-networks can be identified which perform close to the original DNN \cite{gale2019state}.
This observation has motivated the emergence of a host of different approaches to identify such sub-networks 
\cite{hoefler2021sparsity}.
Despite the very promising empirical results, the ultimate performance of DNN compression is yet to be characterized. 
%
This problem can be summarized as the following question:
\begin{displayquote}
\emph{
 What is the minimum performance loss that can be guaranteed when compressing any network of $N$ nodes into a smaller network of $M<N$ nodes?}
\end{displayquote}

\noindent
{\bf Main contributions.} In this paper, we address the above question
by investigating the 
limiting performance when approximating
a two-layer \emph{target} network with $N$ neurons via a two-layer \emph{compressed} network with $M$ neurons.
More precisely, we assume that \emph{(i)} the input features are Gaussian, and that \emph{(ii)} the network weights are i.i.d. 
sub-Gaussian.
Then, we consider the problem of determining the parameters of the compressed network which minimize the expected $L_2$ distance (population loss) between the outputs of the target and compressed network. 
This optimization can be divided into the two sub-problems: that of determining the \emph{amplitudes} and the \emph{directions} of the weights. 
While finding the amplitudes boils down to linear regression, finding the directions 
is  challenging.
Thus, by using tools from high-dimensional probability, we show that this latter problem can be simplified in the asymptotic regime in which $N, d\gg 1$.
We highlight that the simplified objective and the optimal parameters of the pruned network do not depend on the realization of the target network.
Furthermore, if the activation function is ReLU, we formulate the \emph{Equiangular Tight Frame (ETF) conjecture}, that is, we conjecture that the optimal choice of weights on the compressed network forms an ETF (see Fig. \ref{fig:sphere}), while  the orientation of the ETF and the scaling of the neuronal outputs depends on the 
parameters of the target network. We also provide numerical evidence supporting our conjecture.

\noindent {\bf Related work.} 
%
Early work used the second order Taylor approximation of the cost function to determine the relevance of the network weights \cite{lecun1989optimal,hassibi1992second}. 
Using this change as a measure of salience, weights corresponding to low salience are pruned from the network. 
The magnitude of the weights 
has also been considered as a measure of saliency \cite{han2015learning,see2016compression,narang2017exploring}.
The main advantage of these methods relies on their
computational efficiency. 
The so called ``lottery ticket hypothesis'' suggests that pruning should be performed in a structured manner \cite{frankle2018lottery},
and sparsification techniques which promote a structured form of sparsity have been considered e.g. in \cite{han2016eie}. 
In \cite{ye2020good}, the authors consider a ``greedy forward selection'' strategy in which 
the sub-network is obtained by successively adding (with replacement) neurons to an empty network. In \cite{kim2020neuron}, pruning is performed via
``neuron merging'', i.e., similar neurons are identified by the cosine similarity of their weights and they are successively merged.

\noindent
{\bf Notation.}
We  adopt the short-hands  $[m:n] \triangleq \{m, \ldots, n\}$
 and  $[n] \triangleq [1:n]$. 
 We let $\one$ be the all-1 vector, and $\Iv$ the identity matrix. 
%
%
%
%
Given a vector $v$, $\|v\|_2$ denotes its $L_2$ norm.
Given a matrix $M$, $\|M\|_{\rm op}$ denotes its operator norm, $\|M\|_{\Fsf}$ the Frobenius norm,  and $M^+$ its 
Moore-Penrose pseudo-inverse. 
The unit sphere of dimension $d$ is indicated as $\Scal^{d-1}$. The Gaussian distribution with mean $\muv$ and covariance matrix $\Cv$ is denoted as
$\Ncal(\muv,\Cv)$. Given a sub-Gaussian random variable $w$, its sub-Gaussian norm is defined as $\|w\|_{\psi_2} = \inf \{ t>0 \,\,: \,\,\mathbb E[\exp(w^2/t^2)] \le 2 \}$. Given a random vector $\wv\in\mathbb R^d$, its sub-Gaussian norm is defined as $\|\wv\|_{\psi_2} := \sup_{\uv \,:\, \|\uv\|_2=1} \|\uv^\intercal \wv\|_{\psi_2}$.
%


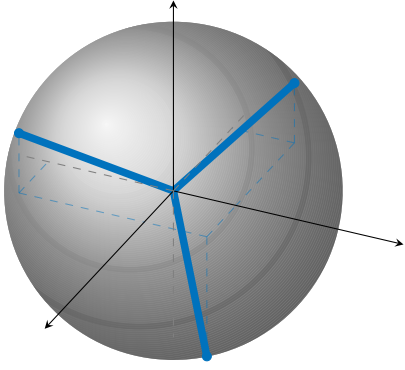
\begin{figure}
\centering  
\begin{tikzpicture}[tdplot_main_coords, scale = 2.25]

\definecolor{mycolor3}{rgb}{0.00000,0.49804,0.00000}%
\definecolor{vamsired}{rgb}{0.63529,0.07843,0.18431} 
\definecolor{vamsiblue}{rgb}{0.00000,0.44706,0.74118} 
\definecolor{vamsigreen}{rgb}{0.00000,0.49804,0.00000} 
\definecolor{vamsiorange}{rgb}{0.87059,0.49020,0.00000} 
\definecolor{mycolor5}{rgb}{0.00000,0.44700,0.74100} %
\definecolor{mycolor6}{rgb}{0.74902,0.00000,0.74902} %

\shade[ball color = lightgray,
    opacity = 0.5
] (0,0,0) circle (1cm);




\coordinate (P1) at  (-0.81649658, 0.40824829, 0.40824829);
\coordinate (P2) at  (0.40824829, - 0.81649658, 0.40824829);
\coordinate (P3) at  (0.40824829, 0.40824829,  -0.81649658);

\draw[fill = vamsiblue, vamsiblue ] (P1) circle (0.75pt);
\draw[fill = vamsiblue,vamsiblue] (P2) circle (0.75pt);
\draw[fill = vamsiblue,vamsiblue] (P3) circle (0.75pt);

\draw[vamsiblue, line cap=round, line width = 1mm, line join=round, ] (0,0,0) -- (P1);
\draw[vamsiblue, line cap=round, line width = 1mm ] (0,0,0) -- (P2);
\draw[vamsiblue, line cap=round, line width = 1mm  ] (0,0,0) -- (P3);

\draw[thin, dashed,vamsiblue, opacity=0.5 ] (-0.81649658, 0.40824829, 0) --  (-0.816 49658, 0, 0);
\draw[thin, dashed,vamsiblue, opacity=0.5] (-0.81649658, 0.40824829, 0) --  (0, 0.40824829, 0);
\draw[thin, dashed,vamsiblue, opacity=0.5] (-0.81649658, 0.40824829, 0) -- (P1);

\draw[thin, dashed,vamsiblue, opacity=0.5 ] (0.40824829, - 0.81649658, 0) --  (0.40824829, 0, 0);
\draw[thin, dashed,vamsiblue, opacity=0.5] (0.40824829, - 0.81649658, 0) --  (0, - 0.81649658, 0);
\draw[thin, dashed,vamsiblue, opacity=0.5] (0.40824829, - 0.81649658, 0) -- (P2);

\draw[thin, dashed,vamsiblue, opacity=0.5 ] (0.40824829, 0.40824829, 0) --  (0.40824829, 0, 0);
\draw[thin, dashed,vamsiblue, opacity=0.5] (0.40824829, 0.40824829, 0) --  
(0, 0.40824829, 0);
\draw[thin, dashed,vamsiblue, opacity=0.5] (0.40824829, 0.40824829, 0) -- (P3);


 
\draw[-stealth] (0,0,0) -- (1.80,0,0);
\draw[-stealth] (0,0,0) -- (0,1.50,0);
\draw[-stealth] (0,0,0) -- (0,0,1.30);
\draw[dashed, gray] (0,0,0) -- (-1,0,0);
\draw[dashed, gray] (0,0,0) -- (0,-1,0);
\draw[dashed, gray] (0,0,0) -- (0,0,-1);

\end{tikzpicture}





 
 
    \caption{A representation of the equiangular tight frame  (ETF).}
    \label{fig:sphere}
    \vspace{-0.5cm}
\end{figure}

\section{Problem formulation}
\label{sec:Problem Formulation}

Let the \emph{target} network $\Wsf$ be specified by \emph{(i)} a set of \emph{weight vectors} $\{ \wv_n \}_{n\in [N]}$  with $\wv_n \in \Scal^{d-1}$, and \emph{(ii)} a set of \emph{scaling coefficients} $\{a_n\}_{n\in [N]}$ with $a_n \in \Rbb$.
%
Then, the output $f_{\Wsf}(\xv)$ is obtained as 
\ea{
f_{\Wsf}(\xv) = \f 1 N \sum_{n \in [N]} a_n \sigma( 
\langle \wv_n, \xv \rangle
),
\label{eq:out original}
}
where $\sigma$ is the activation function.
%
The \emph{compressed} network $\Vsf$ is specified by the weights $\{\vv_m \}_{m\in [M]}$ with $\vv_m \in \Scal^{d-1}$ and scaling coefficients  $\{b_m \}_{m\in [M]}$, so that the output $f_{\Vsf}(\xv)$ is obtained as
\ea{
f_{\Vsf}(\xv) = \f 1 M \sum_{m \in [M]} b_m \sigma(
\langle \vv_m, \xv \rangle
).
\label{eq:pruned dnn}
}
As we allow for the coefficients $\{a_n\}_{n\in [N]}$ and $\{b_m \}_{m\in [M]}$, the assumption that $\wv_n, \vv_m \in \Scal^{d-1}$ is not particularly restrictive. Furthermore, if we assume that the activation function is ReLU ($\sigma(x)=\max(x, 0)$) then, by using the homogeneity of $\sigma$, one can incorporate the norm of the weights $\wv_n$ and $\vv_m$ into the coefficients $a_n$ and $b_m$, respectively.

Our goal is to compress \eqref{eq:out original} into \eqref{eq:pruned dnn} so that the expected $L_2$ loss between the two outputs when the input is Gaussian, i.e., 
\ea{
L(\Vsf,\Wsf) =  \Ebb_{\Xv\sim \Ncal(\zerov, \Iv)}  \lsb \|  f_{\Wsf}(\Xv)-f_{\Vsf}(\Xv) \|_2^2  \rsb,
\label{eq:arg inf}
}
is minimized over the choice of $\{b_m, \vv_m \}_{m\in [M]}$. We remark that, in the theoretical literature, it is common to assume that the input data is Gaussian, see e.g.~\cite[Remark 2]{marcotensor}.
Furthermore, the average $L_2$ distance of two functions on a Gaussian input, as in~\eqref{eq:arg inf}, can be regarded as a proxy for understanding how close they behave on normalized (i.e. whitened) data.


A compact formulation of the expression in \eqref{eq:arg inf} using vector notation can be achieved by defining
\eas{
g(\uv,\vv) &= \Ebb_{\Xv \sim \normal(\zerov, \Iv)} [\sigma(\uv^{\intercal} \Xv) \sigma(\vv^{\intercal} \Xv)],\label{eq:define g}
\\
\lsb G_{WW} \rsb_{ij} & = g(\wv_i, \wv_j), \quad \lsb G_{VV} \rsb_{ij}  = g(\vv_i, \vv_j) ,\\
 \lsb G_{VW} \rsb_{ij}  & = g (\vv_i, \wv_j).
}{\label{eq:vec}}
Similarly, we define the column vectors $\av$, $\bv$ as  $[\av]_n=a_n$ $[\bv]_m=b_m$.
Finally, defining 
$\sv := \frac{1}{N} G_{VW} \av$,
\eqref{eq:arg inf} becomes
\ea{
L(\Vsf,\Wsf) &= \f 1 {N^2} \av^{\intercal} G_{WW} \av- \f 2 {M} \bv^{\intercal} \sv 
+ \f 1 {M^2} \bv^{\intercal} G_{VV} \bv.
\label{eq:vec loss}
 }
The entries of $G_{VV}, G_{WW}$, and $G_{VW}$ can be interpreted as correlations between neurons, and $[\sv]_i$ can be interpreted as the average correlation between the $i^{\rm th}$ neuron in the compressed network and all neurons in the target network.
%
In this problem definition, we use the term ``compression'' with the following acceptation.
The network $\Wsf$ has been trained to perform an inference task over a given dataset.
The problem at hand is to produce a smaller network $\Vsf$  approximating $\Wsf$ using only the knowledge of the network parameters.
In the remainder of the paper, we consider the following set of assumptions: 
\begin{description}
\item \underline{Assumption (A-$\xv$):} the network input $\Xv$  is Gaussian with zero mean and identity covariance. 
\item \underline{Assumption (A-$\wv$):} the target network weights $\{ \wv_n \}_{n\in [N]}$  are i.i.d. zero-mean, unit-norm (i.e., $\wv_n\in \Scal^{d-1}$) sub-Gaussian random vectors, with a sub-Gaussian norm $\|\wv_n\|_{\psi_2}\le\sigma_W/\sqrt{d}$.\footnote{The scaling $1/\sqrt{d}$ ensures that the weights belong to the unit sphere $\Scal^{d-1}$.}
\item \underline{Assumption (A-$\av$):} the scaling coefficients of the target network $\{a_n\}_{n\in [N]}$ are i.i.d. bounded random variables, i.e. $|a_n| \leq A$, and they are independent of the weights $\{ \wv_n \}_{n\in [N]}$. We also assume $\Ebb[a_n] = \mu_a \neq 0$. 
\item \underline{Assumption (A-$\bv$):} the scaling coefficients $\{b_m\}_{m\in [M]}$ are optimized in a compact set, i.e. $|b_m| \leq B$. 
\end{description}
%
 %
 
 A few remarks on the assumptions are in order. First, the quantities $\sigma_W, \mu_a, A$ and $B$ are all numerical constants (independent of $d, N, M$). Second, the assumption that $\{ \wv_n \}_{n\in [N]}$  are i.i.d. is compatible with gradient descent training in the mean-field regime \cite{mei2018mean,rotskoff2018neural,chizat2018global,sirigano2020mean}. In fact, in this regime, the weights obtained from stochastic gradient descent are close to $N$ i.i.d. particles distributed according to the solution of a certain non-linear partial differential equation \cite{mei2018mean, mei2019mean}. This perspective has been used to prove a number of properties of neural networks, including convergence \cite{mei2018mean,chizat2018global,javanmard2020analysis}, mode connectivity \cite{shevchenko2020landscape} and implicit bias \cite{shevchenko2021mean}. 

%
%


\section{Optimization of scaling coefficients}
\label{sec:Simplification of the Pruning Risk}

We begin 
by arguing that the minimization over $\{b_m\}_{m\in [M]}$ in \eqref{eq:arg inf} 
for a given $\{\vv_m\}_{m\in [M]}$ has a rather simple closed-form expression (Lemma \ref{lem:optimal be}).
%
Next, we show that the optimization of the compressed weights $\{\vv_m\}_{m \in [M]}$ 
can be approached using the Hermite expansion of $\sigma$, when the input is Gaussian. 

%
%
\begin{lemma}[Minimization over $\{b_m\}_{m\in [M]}$]
\label{lem:optimal be}
The optimal value of $\bv$ in  \eqref{eq:vec loss} for fixed $\{ \vv_m \}_{m\in [M]}$ is 
\ea{
\bv^*(G_{VV})= \f M {N} G^{+}_{VV} G_{VW} \av=M G^{+}_{VV} \sv.
\label{eq:optimal b when g fixed}
}
Furthermore, when replacing the solution of \eqref{eq:optimal b when g fixed} into  \eqref{eq:vec loss}, the loss is given by
\ea{
 \lnone L(\Vsf,\Wsf) \rabs_{\bv=\eqref{eq:optimal b when g fixed}} & \label{eq:vec loss simple}
= \f 1 {N^2} \av^{\intercal} G_{WW} \av - 
\sv^{\intercal}
G^{+}_{VV} \sv.
}
\end{lemma}

\begin{proof}
Given any set of unit-norm vectors $\{ \vv_m \}_{m\in [M]}$ in~\eqref{eq:vec loss}, the loss is a quadratic (hence convex) program in $\{ b_m \}_{m\in [M]}$.
Thus, the optimal value can be found by setting the gradient $\nabla_{\bv} L$ to $0$, which readily gives \eqref{eq:optimal b when g fixed}.
%
%
By plugging \eqref{eq:optimal b when g fixed} into \eqref{eq:vec loss}, we also obtain \eqref{eq:vec loss simple}. 
\end{proof}
%
%
%
We note that, since by definition the loss is $\ge 0$, $\sv$ always lies in the image of $G_{VV}$, justifying the use of the Moore-Penrose pseudo-inverse.
In fact, if there exists $\kv \in \text{kernel}(G_{VV})$ such that $\langle \kv, \sv \rangle \neq 0$, then the loss evaluated at $\bv = \eta \kv$ is 
\[
 \frac{1}{N^2} \av^{\intercal} G_{WW} \av - \eta \frac{2}{M} \kv^{\intercal} \sv + \frac{\eta^2}{M^2} \kv^{\intercal} G_{VV} \kv = C - C^{'} \eta,
\]
which can be taken to $\pm \infty$ by taking $\eta$ to $\pm \infty$ depending on the sign of $C^{'}$. 

\begin{lemma}[\cite{dutordoir2021deep}]
\label{lem:expansion}
The function $g(\uv,\vv)$ 
in \eqref{eq:define g} is a function of $\alpha=\langle \uv, \vv \rangle$, which can be
expressed as 
\ea{
g(\uv,\vv)  = g(\alpha)= \f 1 {2\pi} \lb \sqrt{1-\alpha^2}+ \al (\pi - \arccos(\alpha)\rb.
\label{eq:closed g}
}
%
Additionally, $g(\alpha)$ admits the well-defined Taylor expansion 
\ea{
g(\alpha) & =  \frac{1}{2 \pi} +  
\frac{1}{4} \alpha +
\sum_{k \geq 1 }\frac{1}{2 \pi} \frac{((2k-3)!!)^2}{(2k)!} \alpha^{2k}
\label{eq:expansion}.
}
%
%
\end{lemma}
\begin{IEEEproof}
The expression in \eqref{eq:closed g} follows from 
\cite{dutordoir2021deep}.
\footnote{We note that the expression in \eqref{eq:closed g} has a missing factor of $1/2$, which we corrected.} 
The expression in \eqref{eq:expansion} is derived through some rather straightforward manipulations. Here, the Taylor coefficients of $g$ are the squared Hermite coefficients of $\sigma$, see \cite[Section 11.2]{o2014analysis}.
\end{IEEEproof}

We note that, for any $\sigma\in L^2(e^{-w^2/2})$\footnote{$\sigma\in L^2(e^{-w^2/2})$ means that the integral $\int_{\mathbb R} \sigma^2(w)e^{-w^2/2}{\rm d}t$ is finite, which ensures that the Taylor expansion of $g$ is well-defined.}, $g(\uv,\vv)$ is a function of $\alpha=\langle \uv, \vv \rangle$, as long as the input distribution is invariant under rotations.

\section{
Limit loss and limit objective
}
\label{sec:Concentration Bounds}

By solving the minimization over $\{b_m\}_{m\in [M]}$, we have reduced the problem to optimizing 
\eqref{eq:vec loss simple} over $\{ \vv_m \}_{m\in [M]}$. Note that the vectors $\{ \vv_m \}_{m\in [M]}$ appear both in $G^{+}_{VV}$ and in $\sv$. Next, we show that, as $N, d\gg 1$, the dependence of $\sv$ on $\{ \vv_m \}_{m\in [M]}$ can be essentially dropped. More specifically, Theorem \ref{lem:approx sv} proves that the loss \eqref{eq:vec loss} is close to the \emph{limit} loss
\begin{equation}
\tilde{L}(\Vsf,\Wsf) =
\f 1 {N^2} \av^{\intercal} G_{WW} \av- \f 2 {M}\f {\mu_a}{2\pi} \bv^{\intercal} \one 
+ \f 1 {M^2} \bv^{\intercal} G_{VV} \bv,
\label{eq:limit loss}
\end{equation}
obtained by plugging $\sv = \frac{\mu_a}{2 \pi}\one$ into \eqref{eq:vec loss}.
To prove the claim, the idea is to exploit the Taylor expansion \eqref{eq:expansion} of $g$, and show that only the zero-th order term has an impact on $\sv$ (as $N, d\gg 1$). To do so, we begin by deriving a bound on the magnitude of the linear term coming from the expansion in \eqref{eq:expansion}.

\begin{lemma}[Controlling the linear term]
\label{lem:linear app}
Let the assumptions (A-$\wv$)-(A-$\av$) hold. Then, we have that 
\[
\sup_{\vv \in \Scal^{d-1}} \labs \frac{1}{N} \sum_{n \in [N]}  a_n \langle \vv, \wv_n \rangle\rabs \leq \frac{t A\sigma_W}{\sqrt{N}},
\]
with probability at least $1 - 2 \exp(-t^2/16)$.
\end{lemma}

\begin{proof}
We first note that
\[
\sup_{\vv \in \Scal^{d-1}} \labs \frac{1}{N} \sum_{n\in [N]} a_n \langle \vv, \wv_n \rangle \rabs = \frac{1}{N} \left\| \dsum_{n\in [N]} a_n \wv_n \right\|_2.
\] 
Since $\{\wv_n\}_{n\in [N]}$ are zero-mean sub-Gaussian vectors with norm $\sigma_W/\sqrt{d}$, $\wv_{\rm avg} := \f 1 N \sum_{n\in [N]} a_n \wv_n$ is also a sub-Gaussian random vector with sub-Gaussian norm $|| \wv_{\rm avg}||_{\psi_2} \leq A \sigma_W/(\sqrt{N d})$.
This can be seen by the following bound on the moment generating function: for any $\uv \in \Scal^{d-1}$, it holds that, for any $\lambda\in\mathbb R$, 
\ea{
\Ebb \lsb  e^{\lambda\uv^{\intercal} \wv_{\rm avg}} \rsb 
& = \prod_{n\in [N]} \Ebb_{a_n, \wv_n} \lsb e^{\lambda\f 1 N a_n \uv^{\intercal} \wv_n}   \rsb \\
& \leq  \prod_{n\in [N]} \Ebb_{a_n} \lsb e^{\lambda^2\f {a^2_n} {N^2} \f {\sigma^2_W} {d} }  \rsb \leq e^{\lambda^2\f {A^2} {N} \f {\sigma^2_W} {d}},
\nonumber
}
where in the first inequality we use $|| \wv_n||_{\psi_2} \leq \sigma_W/\sqrt{d}$ and Proposition 2.5.2 in \cite{vershynin2018high}, and in the second inequality we use that $|a_n|\le A$.
By \cite[Lemma 1]{jin2019short}, we have that
\ea{
\Prob \lsb || \wv_{\rm avg}||_2 \geq \f {tA\sigma_W} {\sqrt{N}} \rsb \leq 2 e^{-\frac{t^2}{16}},
}
which gives the desired result.
\end{proof}

%
Next, we bound the quadratic term of the expansion in \eqref{eq:expansion}.

\begin{lemma} [Controlling the quadratic term]
\label{lem:sum squared}
Let the assumptions (A-$\wv$)-(A-$\av$) hold. Then, we have that 
\begin{equation}
\begin{split}
\sup_{\vv \in \Scal^{d-1}}& \f 1 N \sum_{n\in [N]} |a_n| \langle \vv, \wv_n \rangle^{2} \le  C A \lb\f {\sigma_W^2} {d} + \f 1 N +\f {1} {\sqrt{Nd}}\rb,
\label{eq:sum squared}
\end{split}
\end{equation}
with probability at least $1 - 2 \exp(-d)$, where $C$ is a numerical constant.

\end{lemma}

\begin{proof}
Manipulating the LHS of \eqref{eq:sum squared}, we have
\begin{equation}
    \begin{split}
 &\sup_{\vv \in \Scal^{d-1}}\f 1 N \sum_{n\in [N]} |a_n| \langle \vv, \wv_n \rangle^{2} \le \sup_{\vv \in \Scal^{d-1}}\f 1 N \sum_{n\in [N]} A \langle \vv, \wv_n \rangle^{2} \\
   & \hspace{1em}\leq\sup_{\vv \in \Scal^{d-1}} A \labs \f 1 N \vv^{\intercal} \sum_{n\in [N]}\wv_n \wv_n^{\intercal} \vv \rabs   \leq \f A d \left\|\f 1 N\Wv\Wv^\intercal\right\|_{\rm op},
\end{split}
\end{equation}
where in the first inequality we use that $|a_n|\le A$, and in the last line we have defined $\Wv := [\sqrt{d}\wv_1, \ldots, \sqrt{d}\wv_N]$. 

To bound the operator norm of $\f 1 N \Wv\Wv^\intercal$, let $\boldsymbol \Sigma$ be the second moment matrix of $\sqrt{d}\wv_n$, i.e., $\boldsymbol \Sigma:= d\mathbb E[\wv_n\wv_n^\intercal]$. Then, by  \cite[Remark 5.40]{vershynin2011introduction}, we have that 
\begin{equation}\label{eq:bds0}
    \left\|\f 1 N\Wv\Wv^\intercal\right\|_{\rm op}\le \left\|\boldsymbol \Sigma\right\|_{\rm op}+\max(\delta, \delta^2),
\end{equation}
with probability at least $1-2\exp(-c_1 t^2)$, where $\delta= C_1\sqrt{d/N}+t/\sqrt{N}$ and $c_1, C_1$ are numerical constants. To bound $\left\|\boldsymbol \Sigma\right\|_{\rm op}$, let $\bbu$ be the unitary eigenvector associated to the maximum eigenvalue of $\boldsymbol \Sigma$. Then,
	\begin{equation}\label{eq:int1}
		\left\|\boldsymbol \Sigma\right\|_{\rm op} = \bbu^\intercal \boldsymbol \Sigma \bbu = d\,\mathbb E \left[ \bbu^\intercal \wv_n\wv_n^\intercal \bbu \right]=d\,\mathbb E \left[ (\bbu^\intercal \wv_n)^2 \right].
	\end{equation}
	Furthermore, we have that
	\begin{equation}\label{eq:int2}
	\begin{split}
		\|\wv_n\|_{\psi_2} &:= \sup_{\bbu' \,:\, \|\bbu'\|_2 = 1} \|(\bbu')^\intercal \wv_n\|_{\psi_2}  \\
		&\geq \| \bbu^\intercal \wv_n\|_{\psi_2}\geq \frac{1}{C_2} \sqrt{ \mathbb E \left[ (\bbu^\top \wv_n)^2 \right]},
	\end{split}
	\end{equation}
	where $C_2$ is a numerical constant, and the last inequality comes from Eq. (2.15) of \cite{vershynin2018high}. By combining \eqref{eq:int1} and \eqref{eq:int2} with the fact that $\|\wv_n\|_{\psi_2}\le \sigma_W/\sqrt{d}$, we conclude that 
	\begin{equation}\label{eq:bds}
	    \left\|\boldsymbol \Sigma\right\|_{\rm op}\le C_2^2\sigma_W^2.
	\end{equation}
Hence, the desired result follows by taking $t=d/\sqrt{c_1}$ in \eqref{eq:bds0} and using \eqref{eq:bds}.
%
\end{proof}
%
%
%
%
%
%
At this point, we are ready to prove our main result. 

\begin{theorem}[Convergence to limit loss]
\label{lem:approx sv}
Let the assumptions (A-$\wv$)-(A-$\av$) hold. Then, we have that 
\begin{equation}
\begin{split}
\sup_{\substack{\{\vv_m\}_{m \in [M]} \\ \vv_m \in \Scal^{d-1}}}
\labs [\sv]_i - \f {\mu_a} {2 \pi} \rabs  
&\le {\sf err}(N, d, A, \sigma_W, t),
\label{eq:approx sv}
\end{split}
\end{equation}
with probability at least $1 -  4 e^{-t^2/16}-2e^{-d}$, where 
\begin{equation}
\begin{split}
&{\sf err}(N, d, A, \sigma_W, t) 
:=  C A \bigg(\f {\sigma_W^2} {d} + \f {t(1+\sigma_W)}{\sqrt{N}} \bigg),
\label{eq:approx l}
\end{split}
\end{equation}
and $C$ is a numerical constant. Furthermore, by assuming also (A-$\bv$), we have that, with the same probability, 
 \ea{
\sup_{\substack{\{\vv_m\}_{m \in [M]} \\ \vv_m \in \Scal^{d-1}}}
|L(\Vsf,\Wsf)-\tilde{L}(\Vsf,\Wsf)| \le B\,M\,{\sf err}(N, d, A, \sigma_W, t),\label{eq:diffloss}
 }
 where $L(\Vsf,\Wsf)$ and $\tilde{L}(\Vsf,\Wsf)$ are defined in \eqref{eq:vec loss} and \eqref{eq:limit loss}, respectively.
\end{theorem}
\begin{IEEEproof}
By plugging the Taylor expansion of $g$ (cf.  \eqref{eq:expansion}) into the definition $\sv := \frac{1}{N} G_{VW} \av$, we have that, for $i\in [M]$,
\begin{equation}
\begin{split}
[\sv]_i =  \f 1 N \sum_{n \in [N]} \f {a_n} {2 \pi} &+ \f 1 N \sum_{n \in [N]} \f 1 4 a_n \langle \vv_i, \wv_n \rangle \\
&+  
\f 1 N \sum_{n \in [N]} \sum_{k\ge 1} a_n c_k\langle \vv_i, \wv_n \rangle^{2k},
\end{split}
\end{equation}
where $c_k = \frac{1}{2 \pi} \frac{((2k-3)!!)^2}{k!}$.
Hence,
\begin{equation}
    \begin{split}
    \labs [\sv]_i - \f {\mu_a} {2 \pi} \rabs &\leq   \labs \f 1 N \sum_{n \in [N]} \f {a_n} {2 \pi}  -  \f {\mu_a} {2 \pi} \rabs
    + \labs\f 1 N \sum_{n \in [N]} \f 1 4 a_n \langle \vv_i, \wv_n \rangle \rabs\\
    &\hspace{-1em} +  
    \labs \f 1 N \sum_{n \in [N]} \sum_{k \ge 1} a_n c_k\langle \vv_i, \wv_n \rangle^{2k} \rabs:= T_1 +T_2+T_3.
    \end{split}
\end{equation}
Using Hoeffding inequality for bounded variables, we get
\begin{equation}\label{eq:bd1}
    \Prob \lsb T_1 \geq \frac{tA}{2\pi\sqrt{N}} \rsb \leq 2 e^{-\f {t^2} {2}}.
\end{equation}
Furthermore, using Lemma \ref{lem:linear app}, we have 
\begin{equation}\label{eq:bd2}
\Prob \lsb \sup_{\vv_i \in \Scal^{d-1}} T_2 \geq \frac{tA\sigma_W}{4\sqrt{N}} \rsb \leq 2 e^{-\f {t^2} {16}}.
\end{equation}
For the last term, we first note that 
\begin{equation}\label{eq:bd3}
T_3 \le    \f 1 N \sum_{n \in [N]} \sum_{k \ge 1} |a_n| c_k\langle \vv_i, \wv_n \rangle^{2} \le \f C N \sum_{n \in [N]} |a_n| \langle \vv_i, \wv_n \rangle^{2},  
\end{equation}
where $C$ denotes a numerical constant. Here, in the first inequality we use that $c_k\ge 0$ for all $k\ge 1$ and that $|\langle \vv_i, \wv_n \rangle|\le 1$ since $\vv_i, \wv_n \in \Scal^{d-1}$; and in the second equality we use that $\sum_{k \ge 1} c_k$ is finite, since $\sigma\in L^2(e^{-w^2/2})$. Thus, we can use Lemma \ref{lem:sum squared} to give a uniform bound\footnote{It is uniform in the sense that it upper bounds $\sup_{\vv_i \in \Scal^{d-1}} T_3$.} on the RHS of \eqref{eq:bd3}. By combining this with \eqref{eq:bd1} and \eqref{eq:bd2}, after some simplifications we obtain \eqref{eq:approx sv}.

Recall from assumption (A-$\bv$) that $|b_m|\le B$. Then, an immediate consequence of \eqref{eq:approx sv} is that
\begin{equation}
  \sup_{\substack{\{\vv_m\}_{m \in [M]} \\ \vv_m \in \Scal^{d-1}}}
  \labs \bv^{\intercal} \sv - \bv^{\intercal}  \f {\mu_a} {2 \pi}\one\rabs\le M B{\sf err}(N, d, A, \sigma_W, t),
\end{equation}
which readily gives \eqref{eq:diffloss}.
\end{IEEEproof}

Theorem \ref{lem:approx sv} shows that, if we optimize the limit loss \eqref{eq:limit loss} instead of the original loss \eqref{eq:vec loss}, we suffer at most an additive error which scales as $1/\sqrt{N}+1/d$. 
The same result can be obtained for any $\sigma\in L^2(e^{-w^2/2})$, as long as the zero-th, first and second Hermite coefficients are non-zero. If at least one of these coefficients is zero, we still expect convergence to the limit loss, but with a different additive error term.

By the same argument of Lemma \ref{lem:optimal be}, the optimal value of $\bv$ in \eqref{eq:limit loss} for a fixed $G_{VV}$ is 
\ea{
\tilde{\bv}^*(G_{VV})= \f {M\mu_a}{2\pi} G^{+}_{VV} \one.
\label{eq:optimal b when g fixed2}
}
Furthermore, when replacing the solution of \eqref{eq:optimal b when g fixed2} into  \eqref{eq:limit loss}, the loss is given by
\ea{
 \lnone \tilde \Lsf(\Vsf,\Wsf) \rabs_{\bv=\eqref{eq:optimal b when g fixed2}} & \label{eq:vec loss simple2}
= \f 1 {N^2} \av^{\intercal} G_{WW} \av - 
\frac{\mu_a^2}{4\pi^2}\one^{\intercal}
G^{+}_{VV} \one.
}
This is equivalent to maximizing
\begin{equation}
\label{eqn:simpleobj}
\tilde \Rsf(\Vsf) =
%
\one^{\intercal} G_{VV}^{+} \one.
\end{equation}
We refer
to  \eqref{eq:vec loss simple2}
as the \emph{limit loss} and 
to $\tilde \Rsf(\Vsf)$ as the \emph{limit objective}.
Note that, in \eqref{eqn:simpleobj}, we have removed the dependency on the target network $\Wsf$, as the limit objective no longer depend on the target network. 
We also remark that our bound for the approximation gap in \eqref{eq:diffloss} worsens as $M$ increases. This is to be expected since, when $M$ scales e.g. linearly with $N$, the weights of the compressed network should depend on the realization of the target network.


\section{Equiangular Tight Frame (ETF) conjecture}
\label{sec: ETF}
We conjecture that the limit objective \eqref{eqn:simpleobj} is maximized when the weights of $\Vsf$ form an ETF, i.e., 
$\vv_i^\intercal\vv_j =-1/(M-1)$ for all 
$i \neq j$.
More precisely, we conjecture that \emph{any} ETF achieves the maximum in \eqref{eqn:simpleobj}:
optimizing  the orientation of the ETF results in a reduction of the compression loss bounded as in \eqref{eq:diffloss} -- with high probability.
%
%
%
%
%
We indicate the value of \eqref{eqn:simpleobj} when $\{\vv_m\}$ is an ETF as $\tilde \Rsf_{\rm ETF}$ and the corresponding compressed network as $\Vsf_{\rm ETF}$. We note that the assumption (A-$\bv$) is satisfied when $\bv$ is given by \eqref{eq:optimal b when g fixed2} and the weights of $\Vsf$ form an ETF. We compute the value of $\tilde \Rsf_{\rm ETF}$ in Lemma \ref{lem:etf objective}, and prove the conjecture in the $M=2$ case in Lemma \ref{lemma:M2}. 
\begin{lemma}[ETF limiting objective]
\label{lem:etf objective}
When  the weights of $\Vsf$ form an ETF, we have that
\ea{
\tilde \Rsf(\Vsf) = \tilde \Rsf_{\rm ETF}
= \frac{M}{g(1) + (M - 1) g\lb  -\frac{1}{M-1} \rb },
\label{eq:etf objective}
}
where $g(\alpha)$ is given by \eqref{eq:expansion}.
\end{lemma}
\begin{IEEEproof}
Define $c_1=g(-1/(M-1))$ and $c_2= g(1)-g(-1/(M-1))$. Then, $g(\vv_i^\intercal\vv_j)=c_1$ if $i\neq j$, and $g(\vv_i^\intercal\vv_j)=c_1+c_2$ otherwise. Hence, $G_{VV} = c_1 \one\one^\intercal + c_2 \Iv$ and, by using the Sherman-Morrison formula, 
%
\eqref{eq:etf objective} 
follows.
\end{IEEEproof}
%
%
%
%
%

\begin{lemma}[Case $M=2$]\label{lemma:M2}
For $M=2$, an ETF is the optimal solution of the optimization in \eqref{eqn:simpleobj}.
\end{lemma}
%
%
\begin{IEEEproof}
To see this, let us rewrite \eqref{eqn:simpleobj} as 
%
\begin{multline}
2 
\one^{\intercal} 
\lsb \p{
1 & g_{12}\\
g_{12} & 1 }
\rsb^{-1}
\one = 
\frac{4 (1 - g_{12})}{1 - g_{12}^2} = \frac{4}{1 + g_{12}},
\end{multline}
where $g_{12} = 2  g(\langle \vv_1, \vv_2 \rangle)$.
Since $g$ in \eqref{eq:expansion}  is an increasing function in $\al$, the maximum of the limit objective 
happens when $\vv_1 = -\vv_2$, which is an ETF. 
\end{IEEEproof}
In the rest of the section, we provide numerical results giving a rather strong indication that the ETF maximizes 
\eqref{eqn:simpleobj}.
First, in Fig. \ref{fig:accuracy}, we compare \emph{(i)} the value of the limit objective attained by the ETF (cf. Lemma \ref{lem:etf objective}) and \emph{(ii)} the limit objective obtained after running $10^3$ iteration of gradient descent (GD), as a function of $M$. 
%
These two values coincide up to the $12^{\rm th}$ decimal place for $M$ between $5$ and $30$.
%
%
%


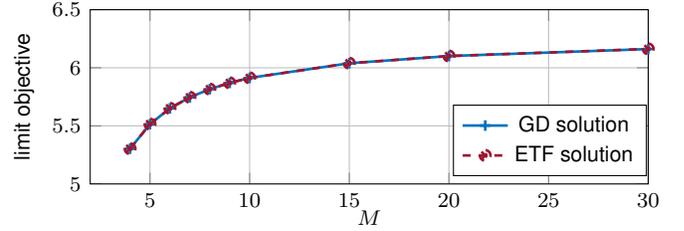
\begin{figure}[t!]
  \centering
 \begin{tikzpicture}[font=\footnotesize\sffamily]
    \definecolor{mycolor1}{rgb}{0.00000,0.44706,0.74118}%
    \definecolor{mycolor2}{rgb}{0.63529,0.07843,0.18431}%
    \definecolor{mycolor3}{rgb}{0.00000,0.49804,0.00000}%
    \begin{axis}[
      height=3.9 cm,
        width=9 cm,
         ymin=5,
         ymax=6.5,
         xmin=2, 
         xmax=30,
         x label style={at={(axis description cs:0.5,-0.1)},anchor=north},
           xmajorgrids={true},
           ymajorgrids={true},
           legend style={at={(0.65,0.25)},anchor=west},
           xlabel={$M$},
           ylabel={limit objective},
        ]
\coordinate (top) at (axis cs:1,\pgfkeysvalueof{/pgfplots/ymax});

\addplot+[mark=+,draw=mycolor1, line width=1.0pt]
    table[col sep=comma]{./data/SGDmaximas.txt};
 
 \addplot+[mark=otimes,dashed, draw=mycolor2, line width=1.0pt]
    table[col sep=comma]{./data/simplexresults.txt};

 
\legend{GD solution, ETF solution}
\end{axis}

\end{tikzpicture}
 \vspace{-0.8cm}
  \caption{Limit objective as estimated through GD and as computed at ETF.}
 \label{fig:accuracy}
 \vspace{-0.3cm}
\end{figure}

%


\begin{figure}[t!]
  \centering
 \begin{tikzpicture}[font=\footnotesize\sffamily]
    \definecolor{mycolor1}{rgb}{0.00000,0.44706,0.74118}%
    \definecolor{mycolor2}{rgb}{0.63529,0.07843,0.18431}%
    \definecolor{mycolor3}{rgb}{0.00000,0.49804,0.00000}%
    \definecolor{mycolor4}{rgb}{0.74902,0.00000,0.74902} %

    \begin{axis}[
      height=4.5 cm,
        width=9 cm,
          xmin=-5, 
          xmax=1000,
         x label style={at={(axis description cs:0.5,-0.1)},anchor=north},
           xmajorgrids={true},
           ymajorgrids={true},
           xlabel={iteration number},
           ylabel={ $d \lb \Vsf_{\sf ETF},\Vsf_{\sf GD} \rb $}
        ]

\coordinate (top) at (axis cs:1,\pgfkeysvalueof{/pgfplots/ymax});

\addplot+[ no marks,draw=mycolor1,line width=1.0pt]
    table[col sep=comma]{./data/traindyn_frobdists_M=5_iter0.txt};
 
\addplot+[ no marks,draw=mycolor2, line width=1.0pt]
    table[col sep=comma]{./data/traindyn_frobdists_M=10_iter0.txt};

\addplot+[ no marks,draw=mycolor3, line width=1.0pt]
    table[col sep=comma]{./data/traindyn_frobdists_M=15_iter0.txt};

\addplot+[ no marks,draw=mycolor4, line width=1.0pt]
    table[col sep=comma]{./data/traindyn_frobdists_M=30_iter0.txt};
 
\legend{$M=5$,$M=10$,$M=15$,$M=30$}
\end{axis}

\end{tikzpicture}
 \vspace{-0.8cm}
  \caption{Distance between the GD solution and the ETF solution, as in \eqref{eq:distance}.}
 \label{fig:frobenius}
 \vspace{-0.5cm}
\end{figure}
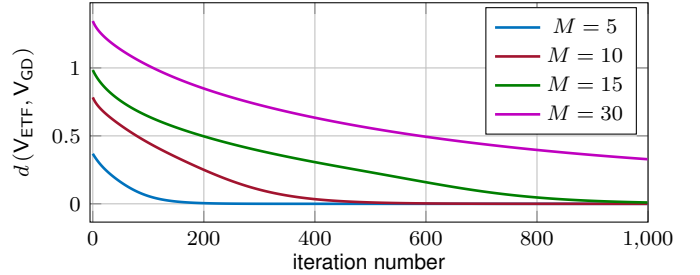

To further investigate the GD solution, let us define the distance between two networks $\Vsf$ and $\Vsf'$ as 
\ea{
 d (\Vsf, \Vsf')=   
 \| 
 G_{VV} -G_{V'V'}
 \|_{\Fsf}.
 \label{eq:distance}
}
In Fig. \ref{fig:frobenius}, we plot $d \lb \Vsf_{\sf ETF},\Vsf_{\sf GD} \rb $, i.e., the distance between the ETF compressed network, $\Vsf_{\sf ETF}$, and  the network compressed using gradient descent, $\Vsf_{\sf GD}$, as a function of the iteration number and for various values of $M$.
We note that 
the convergence of GD to the ETF solution is rather stable and constant with the number of iterations. 
We  also observe that the convergence rate of GD  decreases as $M$ grows large.

\begin{figure}
    \centering

\begin{tikzpicture}[font=\footnotesize\sffamily]
    \definecolor{mycolor1}{rgb}{0.00000,0.44706,0.74118}%
    \definecolor{mycolor2}{rgb}{0.63529,0.07843,0.18431}%
    \definecolor{mycolor3}{rgb}{0.00000,0.49804,0.00000}%
    \begin{axis}[
      height=3.9 cm,
        width=9 cm,
        xmin=-5, 
        xmax=250,
        xmajorgrids={true},
        ymajorgrids={true},
       legend style={at={(0.75,0.25)},anchor=west},
           xlabel={iteration number},
           ylabel={limit objective}
        ]
\coordinate (top) at (axis cs:1,\pgfkeysvalueof{/pgfplots/ymax});


\addplot+[name path=A,draw=mycolor1!50,dotted, no marks,forget plot]
    table[x expr=\coordindex, y=q]{./data/minmax/min_m=2.txt};

\addplot+[name path=B, draw=mycolor1!50,dotted, no marks,forget plot]
    table[x expr=\coordindex, y=q]{./data/minmax/max_m=2.txt};

\addplot[mycolor1!10, forget plot] fill between[of=A and B];

\addplot+[draw=mycolor1, no marks,line legend, line width=1.0pt]
    table[x expr=\coordindex, y=q]{./data/minmax/avg_m=2.txt};


\addplot+[name path=A,draw=mycolor2!50,dotted, no marks,forget plot]
    table[x expr=\coordindex, y=q]{./data/minmax/min_m=4.txt};

\addplot+[name path=B, draw=mycolor2!50,dotted, no marks,forget plot]
    table[x expr=\coordindex, y=q]{./data/minmax/max_m=4.txt};

\addplot[mycolor2!10, forget plot] fill between[of=A and B];

\addplot+[draw=mycolor2, no marks,line legend, line width=1.0pt]
    table[x expr=\coordindex, y=q]{./data/minmax/avg_m=4.txt};


\addplot+[name path=A,draw=mycolor3!50,dotted, no marks,forget plot]
    table[x expr=\coordindex, y=q]{./data/minmax/min_m=8.txt};

\addplot+[name path=B, draw=mycolor3!50,dotted, no marks,forget plot]
    table[x expr=\coordindex, y=q]{./data/minmax/max_m=8.txt};

\addplot[mycolor3!10, forget plot] fill between[of=A and B];

\addplot+[draw=mycolor3, no marks,line legend, line width=1.0pt]
    table[x expr=\coordindex, y=q]{./data/minmax/avg_m=8.txt};

\legend{$M=2$,$M=4$,$M=8$}
\end{axis}

\end{tikzpicture}
    \vspace{-0.8cm}
    \caption{Convergence of 
    GD as a function of the iteration number. 
The boundaries of the shaded area indicate the minimum and maximum of the objective across 10 independent initializations.}
    \vspace{-0.25cm}
    \label{fig:error_plot}
\end{figure}
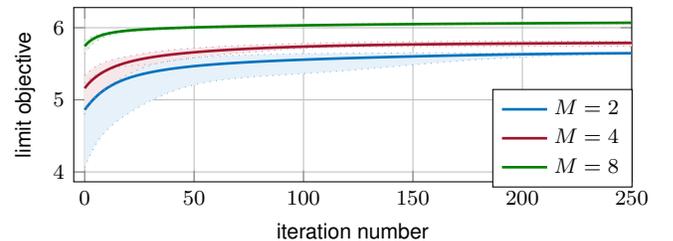

Finally, in Fig. \ref{fig:error_plot}, we plot the variability of the convergence behaviour across multiple GD initializations. 

\section{Conclusion}

In this paper, the compression of a two-layer neural network with $N$ neurons to a network with $M$ neurons is investigated. 
An asymptotic expression for the compression loss is obtained for the case in which the input features are Gaussian and the target weights are i.i.d. sub-Gaussian.
This mean-field view allow to significantly simplify the optimization problem, as the limit objective depends only on $M$ (and not on $N\gg M$).
Numerical simulations suggest that, for a ReLU activation, the optimal structure of the compressed network is the one in which weights lie on an equiangular tight frame,
while the orientation and the scaling of the neurons is obtained from the parameters of the original network. 
Proving this claim is left as an open problem. Interesting avenues for future research also consist in understanding the class of activations for which the ETF conjecture holds and extending the current results to deep networks.


\bibliographystyle{IEEEtran}
\bibliography{sample}

\end{document}